\documentclass[a4paper,11pt]{article}

\usepackage[margin=1in]{geometry}
\usepackage{setspace}
\usepackage{amsmath}
\usepackage{amssymb}
\usepackage{amsthm}
\usepackage{graphicx}
\usepackage{float}
\usepackage{caption}
\usepackage{listings}
\usepackage{tikz}
\usepackage{enumitem}
\usepackage{xcolor}
\usepackage{mathrsfs}

\newtheorem{theorem}{Theorem}[section]

\newtheorem{lemma}[theorem]{Lemma}
\newtheorem{definition}[theorem]{Definition}

\newtheorem{remark}[theorem]{Remark}

\numberwithin{figure}{section}

\providecommand{\keywords}[1]{\textbf{Keywords:} #1}

\usepackage{hyperref}
\hypersetup{
    pdftitle={On Conservative Statistical Riemann Surfaces},
    pdfauthor={Hanwen Liu}
}

\begin{document}

\title{\texorpdfstring{\textbf{On Conservative Statistical Riemann Surfaces}}{On Conservative Statistical Riemann Surfaces}}
\author{Hanwen Liu}
\date{}

\maketitle

\begin{abstract}
We establish a correspondence between information geometry and gauge theory. First, we define an important class of statistical manifolds, that is normalized and satisfies a conservation field equation. Second, we prove that for a conservative statistical structure on an orientable surface, the Chebyshev 1-form is constrained to be harmonic, and the traceless part of the Amari--Chentsov tensor descends to a holomorphic cubic differential. Then, we demonstrate that normalized conservative statistical structures are geometrically generated by solutions to the scalar Tzitz\'eica equation on Higgs bundles with general linear holonomy, generalizing the Labourie-Loftin correspondence. Finally, we prove that the moduli space of normalized conservative statistical structures on a closed orientable surface of genus at least 2 is completely parameterized by a holomorphic vector bundle over the Teichm\"uller space, consisting of Abelian differentials and cubic differentials.
\end{abstract}

\begin{center}
\end{center}

\onehalfspacing
\raggedbottom

\section{Introduction and Background}

The interplay between differential geometry and mathematical statistics has culminated in the rich field of information geometry. A foundational concept in this domain is the statistical manifold, characterized by a Riemannian metric and a totally symmetric rank-3 tensor known as the Amari--Chentsov tensor \cite{Amari1985}. While these structures originally emerged from the study of parametric probability distributions, their geometric properties possess profound connections to affine differential geometry \cite{Nomizu1994} and projective structures \cite{Labourie2007}.

A striking manifestation of geometric rigidity arises when specific differential constraints are imposed on the statistical structure. For instance, the condition that a statistical manifold is divergence-free or possesses vanishing statistical curvature often forces the underlying geometry to completely collapse into well-studied symmetric spaces \cite{Opozda2004}. In accordance with the spirit of rigidity, this paper investigates how global analytic and topological structures on Riemann surfaces can be enforced by a partial differential constraint known as the conservative condition. 

We divide our exploration into two distinct but conceptually related results concerning the rigidity of statistical surfaces. 

In Section 2, we explore the local geometric regularity of statistical structures. Without assuming global symmetries, the Amari--Chentsov tensor is not guaranteed to be compatible with the underlying complex structure of the surface. However, we demonstrate that imposing the conservative condition restricts the divergence of the tensor, rigidly forcing the trace components to be harmonic and the remaining $(3,0)$-components to be holomorphic.

In Section 3, we turn our attention to the global realization of these structures. While the local holomorphicity provides algebraic data, elevating this to a global parametrization requires navigating the non-linear Hitchin equations \cite{Hitchin1987}. We demonstrate that by anchoring the statistical structure to a specific normalization condition, we can exploit the uniformization of Higgs bundles \cite{Hitchin1992, Simpson1988}. This gauge-theoretic stabilization explicitly forces the statistical metrics to solve the scalar Tzitz\'eica equation, a rigidity result we ultimately apply to construct a surjective mapping from the space of normalized conservative statistical structures onto a vector bundle over the Teichm\"uller space of the surface.

\section{Preliminaries on Information Geometry}

We begin by investigating the local differential properties of statistical structures. A statistical manifold is equipped with a metric and a totally symmetric tensor. By controlling the divergence of this tensor via the contracted Bianchi identity, we can extract harmonic and holomorphic data.

\begin{definition}
For a Riemannian manifold $(M,g)$ and a totally symmetric $(0,3)$-tensor $C$ on $M$, the pair $(C,g)$ is called a statistical structure on $M$.
\end{definition}

\begin{definition}
A statistical structure $(C,g)$ on a differentiable manifold $M$ is said to be conservative, if it satisfies the contracted Bianchi identity
\begin{equation}\label{field_equation}
\nabla^g(\operatorname{tr}_g(C))=2\operatorname{div}_g(C)
\end{equation}
where $\nabla^g$ is the Levi-Civita connection of $(M,g)$.
\end{definition}

We first show that the conservative condition rigidly enforces harmonicity on the trace of the statistical structure, naturally generalizing the concept of incompressible flows to information geometry.

\begin{lemma}\label{harmonicity}
Let $(C,g)$ be a conservative statistical structure on a differentiable manifold $M$. Then, the Chebyshev 1-form $\operatorname{tr}_g(C)$ is harmonic.
\end{lemma}
\begin{proof}
Let $\tau := \operatorname{tr}_g(C)$ and let $C_0$ be the traceless part of $C$. Let $n$ be the dimension of $M$.
The tensor $C$ decomposes as
\begin{equation}\label{decomposition_n}
C(u,v,w) = C_0(u,v,w) + \frac{1}{n+2} ( \tau(u)g(v,w) + \tau(v)g(w,u) + \tau(w)g(u,v) ).
\end{equation}
Taking the divergence of equation~(\ref{decomposition_n}) yields
$$\operatorname{div}_g(C)(u,v) = \operatorname{div}_g(C_0)(u,v) + \frac{1}{n+2} (\operatorname{div}_g(\tau)g(u,v) + (\nabla^g_u \tau)(v) + (\nabla^g_v \tau)(u)).$$
Evaluating the trace of the field equation $\nabla^g \tau = 2\operatorname{div}_g(C)$ gives
$$\operatorname{div}_g(\tau) = 2\operatorname{tr}_g(\operatorname{div}_g(C)) = 2\operatorname{div}_g(\tau),$$
which yields $\operatorname{div}_g(\tau) = 0$.
Furthermore, since the divergence $\operatorname{div}_g(C)$ is symmetric in its remaining two indices, the field equation $\nabla^g \tau = 2\operatorname{div}_g(C)$ requires that $\nabla^g \tau$ is symmetric.
The symmetry of $\nabla^g \tau$ is equivalent to the condition $d\tau = 0$.
Because $\tau$ is both closed and divergence-free, it is harmonic.
\end{proof}

To subsequently connect these structures to global moduli spaces, we must impose a scalar constraint linking the norm of the tensor to the underlying curvature of the manifold. This ensures that the induced geometry is strictly hyperbolic in nature.

\begin{definition}
A statistical structure $(C,g)$ on a differentiable manifold $M$ is said to be normalized if 
$$
|C_0|^2=4S_g+16,
$$
where $C_0$ is the traceless part of $C$, and $S_g$ is the Ricci scalar of $(M,g)$.
\end{definition}

We now demonstrate that in complex dimension one, the traceless part of a conservative statistical structure naturally descends to a holomorphic cubic differential. This algebraic extraction is the fundamental bridge between statistical structures and complex analysis.

\begin{lemma}\label{holomorphicity}
Let $(C,g)$ be a conservative statistical structure on an orientable surface $M$, and let $J$ be the complex structure induced by the conformal class of $g$, so that $X=(M,J)$ is a Riemann surface.
Then the $(3,0)$-component of $C$ is a holomorphic cubic differential on $X$.
\end{lemma}
\begin{proof}
Let $\tau := \operatorname{tr}_g(C)$ and let $C_0$ be the traceless part of $C$.
As $M$ is a smooth surface, its dimension is $n=2$, and hence the tensor $C$ decomposes as
\begin{equation}\label{decomposition_2}
C(u,v,w) = C_0(u,v,w) + \frac{1}{4} (\tau (u)g(v,w) + \tau (v)g(w,u)+\tau (w)g(u,v)).
\end{equation}
Taking the divergence of equation~(\ref{decomposition_2}) yields
$$\operatorname{div}_g (C)(u,v) = \operatorname{div}_g (C_0)(u,v) + \frac{1}{4}(\operatorname{div}_g (\tau)g(u,v) + (\nabla^g_u \tau)(v)+(\nabla^g_v \tau)(u)).$$
By Lemma \ref{harmonicity}, the Chebyshev $1$-form $\tau$ is harmonic, which yields $\operatorname{div}_g (\tau) = 0$ and $d\tau = 0$.
The condition $d\tau = 0$ implies that $\nabla^g \tau$ is symmetric, simplifying the divergence identity to
$$\operatorname{div}_g (C) = \operatorname{div}_g (C_0) + \frac{1}{2}\nabla^g \tau.$$
Substituting the field equation $\nabla^g \tau = 2\operatorname{div}_g (C)$ eliminates the trace components entirely, leaving $\operatorname{div}_g (C_0) = 0$.
The incompressibility of the totally symmetric, traceless tensor $C_0$ on the orientable surface $M$ is equivalent to the Cauchy-Riemann equation $\bar{\partial} Q = 0$, where $Q$ is the $(3,0)$-component of $C$.
Therefore, $Q$ is a holomorphic cubic differential on $X$.
\end{proof}

\section{The Higgs Bundle Formulation}

Having established the local holomorphicity of the tensor components, we transition to the global formulation using Higgs bundles. Hitchin's equations provide a powerful gauge-theoretic framework to study holomorphic differentials over Riemann surfaces. By utilizing a specific cyclic grading of the Higgs bundle, we can synthesize the metric and the tensor into a single flat connection.

\begin{definition}
Let $(C,g)$ be a conservative statistical structure on a closed orientable surface $M$, and let $J$ be the complex structure induced by the conformal class of $g$, so that $X=(M,J)$ is a compact Riemann surface.
The moduli data of $(C,g)$ is defined to be the pair $$(\omega,Q)\in H^0(X;K_X)\oplus H^0(X;3K_X),$$ where $Q$ is the $(3,0)$-component of $C$, and $\operatorname{tr}_g(C)=16\operatorname{Re}(\omega)$.
\end{definition}

We now prove that any pair of moduli data uniquely determines a normalized conservative statistical structure through the unique harmonic metric of a cyclic Higgs bundle.

\begin{lemma}\label{core_lemma1}
Let $X$ be a compact Riemann surface of genus at least $2$, and $\mathcal{E}:=K_X^{-1}\oplus \mathcal{O}_X\oplus K_X$.
Let $\pi\colon\mathcal{E}\rightarrow K_X^{-1}$ be the canonical projection, and $\xi:=(0,1,0)^T\in H^0(X;\mathcal{E})$.
Then, for any
$(\omega,Q) \in H^0(X; K_X)\oplus H^0(X; 3K_X)$, there exists a unique Hermitian metric $h$ on $X$ such that $H=\operatorname{diag}(h,1,h^{-1})$ solves the Hitchin equation 
\begin{equation}\label{Hitchin}
F_H+[\Phi,\Phi^*]=0
\end{equation}
on the stable Higgs bundle $(\mathcal{E},\Phi)$, where
$$
\Phi:=\begin{pmatrix}
\omega & 1 & 0\\
0 & \omega & 1\\
Q & 0 & \omega
\end{pmatrix}\in H^0(X;\operatorname{End}(\mathcal{E})\otimes K_X),
$$
and $F_H$ is the curvature of the Chern connection $\nabla^H$ of $(\mathcal{E},H)$.
Moreover, the flat connection $D:=\nabla^H+\Phi+\Phi^*$ induces a normalized conservative statistical structure $(-\nabla g,g)$ on the underlying differentiable manifold $M$ of $X$ via $g:=\operatorname{Re}(h)$ and $$\nabla_vw:=\pi(D_vD_w\xi)+2g(v,w)\operatorname{Re}(\omega)^\sharp,$$ where $v,w$ are smooth vector fields on $M$, and $\sharp$ is the musical isomorphism induced by $g$.
\end{lemma}
\begin{proof}
We first show that the unique harmonic metric solving the Hitchin equation has the diagonal form $H = \operatorname{diag}(h, 1, h^{-1})$.
The Higgs field $\Phi$ can be decomposed as $\Phi = \omega I + A$, where $I$ is the identity endomorphism and 
$$
A = \begin{pmatrix} 0 & 1 & 0 \\ 0 & 0 & 1 \\ Q & 0 & 0 \end{pmatrix}\in H^0(X;\operatorname{End}(\mathcal{E})\otimes K_X)
$$
is a cyclic Higgs field.
The Higgs field $\Phi$ is a matrix-valued $1$-form, and the commutator in the Hitchin equation is the graded Lie bracket $[\Phi, \Phi^*] = \Phi \wedge \Phi^* + \Phi^* \wedge \Phi$.
Because $\omega$ and $\bar{\omega}$ are $1$-forms, their wedge products anti-commute, yielding $\omega \wedge \bar{\omega} + \bar{\omega} \wedge \omega = 0$.
Similarly, the cross terms $\omega I \wedge A^* + A^* \wedge \omega I$ vanish.
Therefore, the commutator $[\Phi, \Phi^*]$ reduces to the nilpotent part as
$$
[\Phi, \Phi^*] = [A, A^*].
$$
Assuming the diagonal ansatz $H = \operatorname{diag}(h, 1, h^{-1})$, the metric adjoint $A^* = H^{-1}A^\dagger H$ evaluates to
$$
A^* = \begin{pmatrix} h^{-1} & 0 & 0 \\ 0 & 1 & 0 \\ 0 & 0 & h \end{pmatrix} \begin{pmatrix} 0 & 0 & \bar{Q} \\ 1 & 0 & 0 \\ 0 & 1 & 0 \end{pmatrix} \begin{pmatrix} h & 0 & 0 \\ 0 & 1 & 0 \\ 0 & 0 & h^{-1} \end{pmatrix} = \begin{pmatrix} 0 & 0 & \bar{Q}h^{-2} \\ h & 0 & 0 \\ 0 & h & 0 \end{pmatrix}.
$$
Matrix multiplication yields the diagonal commutator:
$$
[A, A^*] = \operatorname{diag}(h - |Q|^2h^{-2}, \, 0, \, -h +  |Q|^2h^{-2}).
$$
The curvature of the Chern connection for $H$ is $F_H = \operatorname{diag}(F_h, 0, -F_h)$, where $F_h$ is the curvature of $(X,h)$.
Substituting this into the Hitchin equation $F_H + [\Phi, \Phi^*] = 0$ yields a trivial identity for the central $\mathcal{O}_X$ component, and reduces the remaining system to the scalar Tzitz\'eica equation 
\begin{equation}\label{Scalar_Hitchin}
F_h + h - |Q|^2h^{-2} = 0.
\end{equation}
By the existence and uniqueness theorem applied to the cyclic Higgs bundle $(\mathcal{E},A)$, there is a unique Hermitian metric solving equation~(\ref{Scalar_Hitchin}), confirming that $H$ is diagonal and restricts to the Hermitian metric $h$ on $T_X$.
We now study the properties of the connection $\nabla$. Since $H$ is diagonal with middle entry equal to $1$, the Chern connection satisfies $\nabla^H \xi = 0$ for $\xi=(0,1,0)^T$.
The action of the Higgs field on $\xi$ yields that
$$
\Phi(v) \xi = (\omega(v)I + A(v))\xi = \omega(v)\xi + v,
$$
and that
$$
\Phi^*(v) \xi = (\bar{\omega}(v)I + A^*(v))\xi = \bar{\omega}(v)\xi + v^\flat,
$$
where $v$ is a smooth vector field on $X$ and $\flat$ is the musical isomorphism induced by $h$.
Denoting $\Omega := \omega + \bar{\omega}$, the flat connection evaluates to $D_v \xi = (v , \Omega(v) , v^\flat)$.
Therefore, for smooth vector fields $v,w$ on $X$, we have
$D_v D_w \xi = D_v w +D_v (w^\flat)+ D_v(\Omega(w)\xi)$.
Denote by $\nabla^h$ the Chern connection of $(T_X,h)$. Computation yields $\pi(D_v w)=\nabla^h_v w + \Omega(v)w$ and $\pi(D_v(\Omega(w)\xi))=\Omega(w)v$.
Moreover, the projection image of $D_v (w^\flat)$ arises solely from $A^*(v) w^\flat$, which equals $\pi(D_v (w^\flat))=h^{-1}\bar{Q}(v,w)$.
We therefore arrive at
$$\nabla_v w = \pi(D_v D_w \xi) + g(v,w)\Omega^\sharp = \nabla^h_v w + \Omega(v)w + \Omega(w)v + g(v,w)\Omega^\sharp + h^{-1}\bar{Q}(v,w).$$
To verify that $\nabla$ is torsion-free, we examine the torsion tensor $T(v,w) = \nabla_v w - \nabla_w v - [v,w]$ of $\nabla$.
Since the Chern connection $\nabla^h$ is torsion-free, we have $\nabla^h_v w - \nabla^h_w v = [v,w]$.
The terms $\Omega(v)w + \Omega(w)v$ are symmetric in $v$ and $w$.
Furthermore, since $Q \in H^0(X; 3K_X)$ is a totally symmetric cubic form, the bilinear form $h^{-1}\bar{Q}$ is symmetric.
Thus, $T(v,w) = 0$.

Next, we establish that $(-\nabla g, g)$ is a conservative statistical structure for $g:=\operatorname{Re}(h)$.
By the metric compatibility of the Chern connection, the covariant derivative of the Riemannian metric $g$ reduces to contractions with the difference tensor $B(v,w) := \nabla_v w - \nabla^h_v w$.
Evaluating the definition of $\nabla$ yields $$B(v,w) = \Omega(v)w + \Omega(w)v + g(v,w)\Omega^\sharp + h^{-1}\bar{Q}(v,w).$$
Substitution into $(\nabla_u g)(v,w) = -g(B(u,v), w) - g(v, B(u,w))$ yields the equation 
\begin{equation}\label{collapse}
(\nabla_u g)(v,w) = -2\Omega(u)g(v,w) - 2\Omega(v)g(u,w) - 2\Omega(w)g(u,v) - 2\operatorname{Re}(\bar{Q}(u,v,w)).
\end{equation}
Since this expression~(\ref{collapse}) is invariant under permutations of $u$, $v$, and $w$, the third-order tensor $\nabla g$ is totally symmetric.
We now define the Amari--Chentsov tensor by evaluating $C:=-\nabla g$, yielding
$$
C(u,v,w) = -(\nabla_u g)(v,w) = 2\Omega(u)g(v,w) + 2\Omega(v)g(u,w) + 2\Omega(w)g(u,v) + 2\operatorname{Re}(\bar{Q}(u,v,w)).
$$
This confirms that the third-order tensor $C$ is totally symmetric.
The traceless part of $C$ is thus $C_0 = 2\operatorname{Re}(\bar{Q}) = Q + \overline{Q}$.
Therefore, the $(3,0)$-component of $C$ is $Q$, confirming the moduli data.
We verify that the tensor $C$ satisfies the field equation~(\ref{field_equation}).
Taking the trace of $C$ yields the Chebyshev $1$-form $\tau := \operatorname{tr}_{g} (C)$.
Noting that the trace of $\operatorname{Re}(\bar{Q})$ vanishes, we evaluate the trace to obtain $\tau = 8\Omega = 16\operatorname{Re}(\omega)$, confirming the definition of $\omega$.
Concurrently, we evaluate the divergence $\operatorname{div}_{g} (C)$ with respect to the Levi-Civita connection $\nabla^g$.
Because $X$ is a K\"ahler manifold, the Chern connection $\nabla^h$ coincides with the Levi-Civita connection $\nabla^g$ on the tangent bundle.
Since $\omega \in H^0(X; K_X)$ is holomorphic, the $1$-form $\Omega = 2\operatorname{Re}(\omega)$ is harmonic, ensuring it is divergence-free.
Similarly, since $Q \in H^0(X; 3K_X)$ is holomorphic, its associated totally symmetric, traceless tensor $\operatorname{Re}(\bar{Q})$ is incompressible.
Consequently, taking the divergence of $C$ simplifies to $$\operatorname{div}_{g} (C)(v,w) = 2(\nabla^{g}_w \Omega)(v) + 2(\nabla^{g}_v \Omega)(w).$$ Because $\Omega$ is a closed 1-form, its covariant derivative is symmetric, reducing the expression to $\operatorname{div}_{g} (C) = 4\nabla^{g} \Omega$.
We arrive at the identity $$\nabla^{g}(\operatorname{tr}_{g} (C))  = 8\nabla^{g} \Omega = 2\operatorname{div}_{g}(C).$$ 
Consequently, the pair $(C, g)$ is a conservative statistical structure on $X$.

Finally, the metric $h$ satisfies the scalar Hitchin equation $F_h + h - |Q|^2h^{-2} = 0$, which evaluates to the Riemannian curvature condition $|C_0|^2 = 4S_g + 16$, ensuring the structure is normalized.
\end{proof}

\begin{remark}
The significance of Lemma \ref{core_lemma1} lies in its explicit geometric construction. By utilizing the graded structure of the cyclic Higgs bundle, the flat connection $D = \nabla^H + \Phi + \Phi^*$ flawlessly internalizes the totally symmetric Amari-Chentsov tensor into its projection mapping, effectively binding the metric curvature and the statistical distortion into a unified gauge-theoretic object.
\end{remark}

We establish the converse statement, demonstrating that any existing normalized conservative statistical structure can be identically recovered from its moduli data via the Higgs bundle construction.

\begin{lemma}\label{core_lemma2}
Let $M$ be a closed orientable surface of genus at least $2$, and $(C,g)$ a normalized conservative statistical structure on $M$.
Let $J$ be the complex structure induced by the conformal class of $g$, so that $X=(M,J)$ is a compact Riemann surface.
Then, the moduli data $(\omega,Q)$ of $(C,g)$ defines a Higgs field
$$
\Phi:=\begin{pmatrix}
\omega & 1 & 0\\
0 & \omega & 1\\
Q & 0 & \omega
\end{pmatrix}\in H^0(X;\operatorname{End}(\mathcal{E})\otimes K_X)
$$ on $\mathcal{E}:=K^{-1}_X\oplus \mathcal{O}_X\oplus K_X$, such that $g=\operatorname{Re}(h)$ and $C=-\nabla g$, where
\begin{enumerate}
    \item [1.] $h$ is the unique Hermitian metric on $X$ such that $H=\operatorname{diag}(h,1,h^{-1})$ solves the Hitchin equation $F_H+[\Phi,\Phi^*]=0$ on the stable Higgs bundle $(\mathcal{E},\Phi)$,
    \item [2.] $F_H$ is the curvature of the Chern connection $\nabla^H$ of $(\mathcal{E},H)$,
\end{enumerate}
and $\nabla$ is the unique connection on $TM$ satisfying $\nabla_vw=\pi(D_vD_w\xi)+2g(v,w)\operatorname{Re}(\omega)^\sharp$ for all smooth vector fields $v,w$ on $M$, where 
\begin{enumerate}
 \item [1.] $\pi\colon\mathcal{E}\rightarrow K_X^{-1}$ is the canonical projection, and $D:=\nabla^H+\Phi+\Phi^*$,
    \item [2.] $\xi:=(0,1,0)\in H^0(X;\mathcal{E})$, and $\sharp$ is the musical isomorphism induced by $g$.
\end{enumerate}
\end{lemma}
\begin{proof}
Let $(C,g)$ be a normalized conservative statistical structure on $X$.
By Lemma \ref{harmonicity}, the trace form $\tau := \operatorname{tr}_g(C)$ is harmonic.
Since $X$ is a K\"ahler manifold, any real harmonic $1$-form on $X$ decomposes into the sum of a holomorphic and an anti-holomorphic $1$-form.
Therefore, there exists a unique $\omega \in H^0(X; K_X)$ such that $\tau = 16\operatorname{Re}(\omega)$.
By Lemma \ref{holomorphicity}, the $(3,0)$-component $Q$ of $C$ is a holomorphic cubic differential, establishing that $Q \in H^0(X; 3K_X)$.
This establishes the moduli data $(\omega, Q)$.

By hypothesis, the statistical structure is normalized, satisfying the condition $$|C_0|^2 = 4S_g + 16.$$
Let $h$ be the Hermitian metric on $T_X$ such that $g = \operatorname{Re}(h)$.
Rewriting the normalized condition in terms of the Hermitian metric $h$ and the holomorphic differential 
$Q$ yields the scalar Tzitz\'eica equation
\begin{equation}\label{Scalar_Hitchin_2}
F_h + h - |Q|^2h^{-2} = 0,
\end{equation}
where $F_h$ is the curvature of the Chern connection of $(X,h)$.
By the uniformization theorem for Higgs bundles over a compact Riemann surface of genus at least $2$, there exists a unique Hermitian metric $h$ on $T_X$ solving equation~(\ref{Scalar_Hitchin_2}) for the given $Q$.
This unique $h$ defines the diagonal harmonic metric $H = \operatorname{diag}(h,1,h^{-1})$ on $\mathcal{E}$ satisfying $F_H + [\Phi, \Phi^*] = 0$.
Let $\hat{\nabla}$ be the connection on $TX$ defined by $\hat{\nabla}_v w = \pi(D_v D_w \xi) + 2g(v,w)\operatorname{Re}(\omega)^\sharp$.
By Lemma~\ref{core_lemma1}, $\hat{\nabla}$ is a torsion-free connection and $-\hat{\nabla} g = \hat{C}$, where $(\hat{C}, g)$ is a conservative statistical structure with moduli data $(\omega, Q)$.
Because a real, totally symmetric $(0,3)$-tensor on a Riemann surface is uniquely determined by its trace and $(3,0)$-component, and both $C$ and $\hat{C}$ share the trace $16\operatorname{Re}(\omega)$ and $(3,0)$-component $Q$, we obtain $C = \hat{C}$.
Thus, $C = -\hat{\nabla} g$.

To see that $\hat{\nabla}$ is the unique such connection, suppose $\nabla$ is any torsion-free connection satisfying $-\nabla g = C$.
Then, the difference $\beta(u,v) := \nabla_u v - \hat{\nabla}_u v$ is a symmetric tensor.
The condition $\nabla g = \hat{\nabla} g = -C$ implies that the covariant derivatives of the metric coincide, yielding
$$g(\beta(u,v), w) + g(v, \beta(u,w)) = 0$$
for all smooth vector fields $u, v, w$ on $M$.
Permuting the indices cyclically and utilizing the symmetry of $\beta$ gives $g(\beta(u,v), w) = 0$, which implies $\beta = 0$.
Therefore, we conclude that $\nabla = \hat{\nabla}$ is unique.
\end{proof}

With the bijection rigorously established, we arrive at our main global result. We demonstrate that the collection of all normalized conservative statistical structures is elegantly parameterized by the Teichm\"uller space of the surface.

\begin{theorem}\label{correspondence}
Let $M$ be a closed orientable surface of genus $p\geq2$, and let $\mathcal{S}_p$ be the collection of normalized conservative statistical structures on $M$.
Let $\mathcal{T}_p$ be the Teichm\"uller space of genus $p$, and $$\tilde{\mathcal{T}}_p:=\bigcup_{X\in\mathcal{T}_p}H^0(X;K_X)\oplus H^0(X;3K_X).$$
Then, the mapping $\mu\colon\mathcal{S}_p\rightarrow\tilde{\mathcal{T}}_p$ that sends each conservative statistical structure in $\mathcal{S}_p$ to its moduli data is surjective. Moreover, two normalized conservative statistical structures $(C,g)$ and $(C',g')$ share the same moduli data if and only if there exists a diffeomorphism $\phi\in\operatorname{Diff}(M)$ isotopic to the identity such that $(C',g')=(\phi^*C,\phi^*g)$.
\end{theorem}
\begin{proof}
We first establish the surjectivity of the mapping $\mu$.
Let $(\omega, Q) \in \tilde{\mathcal{T}}_p$.
By Lemma \ref{core_lemma1}, there exists a unique Hermitian metric $h$ on $X$ solving the corresponding Hitchin equation, which induces a normalized conservative statistical structure $(C, g)$ on $M$ via $g= \operatorname{Re}(h)$.
The moduli data of this structure is precisely $(\omega, Q)$.
Consequently, the pair $(C, g)$ is a normalized conservative statistical structure, so $(C, g)$ $\in \mathcal{S}_p$.
Thus, we have that $\mu(C, g) = (\omega, Q)$, proving that $\mu$ is surjective.

Next, we establish the equivalence of structures sharing the same moduli data.
Let $(C, g)$ and $(C', g')$ be two normalized conservative statistical structures in $\mathcal{S}_p$ that map to the same moduli data $(\omega, Q) \in \tilde{\mathcal{T}}_p$.
Because the Teichm\"uller space $\mathcal{T}_p$ parametrizes complex structures modulo diffeomorphisms isotopic to the identity, the conformal classes of $g$ and $g'$ are related by a diffeomorphism $\phi \in \operatorname{Diff}(M)$ isotopic to the identity, namely $\phi^*[g]=[g']$.
Pulling back $(C, g)$ via $\phi$ yields a new normalized conservative statistical structure $(\phi^*C, \phi^*g)$ on $M$.
By Lemma \ref{core_lemma2}, on a fixed Riemann surface $X$, the metric $\phi^*g$ is uniquely determined as $\phi^*g = \operatorname{Re}(h)$, where $h$ is the unique Hermitian metric solving the scalar Hitchin equation for the fixed holomorphic differential $Q$.
Since the Hitchin equation in Lemma \ref{core_lemma2} depends only on the complex structure and $Q$, both $g'$ and $\phi^*g$ must equal $\operatorname{Re}(h)$, yielding $g' \equiv \phi^*g$.
Furthermore, a totally symmetric $(0,3)$-tensor on a Riemann surface is uniquely determined by its trace and $(3,0)$-component.
Since $C'$ and $\phi^*C$ both have trace $16\operatorname{Re}(\omega)$ and $(3,0)$-component $Q$ with respect to the identical metric $g'=\phi^*g$, we obtain $C' = \phi^*C$.
Therefore, we finally arrive at $(C', g') = (\phi^*C, \phi^*g)$, proving the desired condition.
\end{proof}

\section{Conclusions and Remarks}

Theorem~\ref{correspondence} illustrates a profound rigidity phenomenon in statistical geometry. The constraint equations that define a conservative statistical structure locally force the traceless components of the Amari--Chentsov tensor into the rigid algebraic framework of holomorphic cubic differentials. By utilizing the global uniformization machinery of Higgs bundles, we explicitly demonstrated that the global space of such statistical metrics fundamentally mirrors the structure of a specific vector bundle over Teichm\"uller space. This provides a rich, gauge-theoretic geometric anchor for mathematical statistics on surfaces, extending classical results on affine spheres \cite{Loftin2001} into the broader context of information geometry.

\section*{Acknowledgement}

The author is deeply grateful to Weiyi Zhang for very useful and inspiring discussions.
The author expresses gratitude to the reviewers for various suggestions.
This research was completed while the author was studying at the Mathematics Institute of the University of Warwick.
The author therefore would like to thank the University of Warwick for its hospitality.
The author is indebted to the Russian mathematical society, especially, to all the professors teaching the math in Moscow program, and most importantly, to the author's supervisor Alexander Petrovich Veselov at Loughborough University, as they cultivated the author's mathematical literacy and maturity.
\section*{Funding Information}

No funding was received to assist with the preparation of this manuscript, and the author did not receive support from any organization for the submitted work.
\section*{Statements and Declarations}

The author certifies that the author has no affiliations with or involvement in any other organization or entity with any financial interest or non-financial interest in the subject matter or materials discussed in this manuscript.
This article is licensed under a Creative Commons Attribution 4.0 International License, which permits use, sharing, adaptation, distribution and reproduction in any medium or format, as long as you give appropriate credit to the original author and the source, provide a link to the Creative Commons licence, and indicate if changes were made.
The images or other third party material in this article are included in the article’s Creative Commons licence, unless indicated otherwise in a credit line to the material.
If material is not included in the article’s Creative Commons licence and your intended use is not permitted by statutory regulation or exceeds the permitted use, you will need to obtain permission directly from the copyright holder.
Data sharing is not applicable to this article as no datasets were generated or analysed during the current study.
The author hereby provides consent for the publication of the manuscript detailed above.

\bibliographystyle{plain}
\bibliography{references}

@article{Amari1985,
  author  = "Amari, Shun-Ichi",
  title   = "Differential-Geometrical Methods in Statistics",
  journal = "Lecture Notes in Statistics",
  volume  = "28",
  publisher = "Springer-Verlag",
  year    = "1985"
}

@article{Nomizu1994,
  author  = "Nomizu, Katsumi and Sasaki, Takeshi",
  title   = "Affine Differential Geometry",
  journal = "Cambridge Tracts in Mathematics",
  volume  = "111",
  publisher = "Cambridge University Press",
  year    = "1994"
}

@article{Labourie2007,
  author  = "Labourie, Fran\c{c}ois",
  title   = "Flat Projective Structures on Surfaces and Cubic Holomorphic Differentials",
  journal = "Pure and Applied Mathematics Quarterly",
  volume  = "3",
  number  = "4",
  pages   = "1057--1099",
  year    = "2007"
}

@article{Opozda2004,
  author  = "Opozda, Barbara",
  title   = "A class of projectively flat surfaces",
  journal = "Mathematische Zeitschrift",
  volume  = "246",
  pages   = "315--332",
  year    = "2004"
}

@article{Hitchin1987,
  author  = "Hitchin, Nigel J.",
  title   = "The self-duality equations on a Riemann surface",
  journal = "Proceedings of the London Mathematical Society",
  volume  = "s3-55",
  number  = "1",
  pages   = "59--126",
  year    = "1987"
}

@article{Hitchin1992,
  author  = "Hitchin, Nigel J.",
  title   = "Lie groups and Teichmüller space",
  journal = "Topology",
  volume  = "31",
  number  = "3",
  pages   = "449--473",
  year    = "1992"
}

@article{Simpson1988,
  author  = "Simpson, Carlos T.",
  title   = "Constructing variations of Hodge structure using Yang-Mills theory and applications to uniformization",
  journal = "Journal of the American Mathematical Society",
  volume  = "1",
  number  = "4",
  pages   = "867--918",
  year    = "1988"
}

@article{Loftin2001,
  author  = "Loftin, John C.",
  title   = "Affine spheres and convex $\mathbb{R}{P}^n$-manifolds",
  journal = "American Journal of Mathematics",
  volume  = "123",
  number  = "2",
  pages   = "255--274",
  year    = "2001"
}

\end{document}